\documentclass[envcountsame]{llncs}
\usepackage{amsmath,amsfonts}
\usepackage[noend, ruled, noline]{algorithm2e}
\usepackage{subcaption}
\usepackage[textsize=scriptsize]{todonotes}
\usepackage{amssymb,latexsym}
\usepackage{graphicx}
\usepackage{pgf,pgfarrows,pgfnodes,pgfautomata,pgfheaps,pgfshade}
\usepackage[hidelinks]{hyperref}
\usepackage{tikz}
\usepackage{color}
\usepackage{paralist}
\usepackage[]{units}
\usepackage{tikz}
\usetikzlibrary{shapes.geometric}
\usetikzlibrary{calc}
\usetikzlibrary{shapes.misc, positioning}

\newcommand{\ejrls}{{{\textsc{LS-PAV}}}}
\newcommand{\pavscore}{{{\mathrm{pav\text{-}sc}}}}

\newcommand{\aver}[1]{{{\mathrm{avs}_{#1}}}}

\title{Optimal Average Satisfaction 
and Extended Justified Representation in Polynomial Time}
\author{Piotr Skowron$^1$ \and Martin Lackner$^2$ \and Edith Elkind$^2$ \and Luis S\'anchez-Fern\'andez$^3$}
\institute{$^1$TU Berlin, Germany \quad $^2$University of Oxford, UK \\ $^3$Universidad Carlos III de Madrid, Spain}

\begin{document}

\maketitle
\thispagestyle{plain}

\begin{abstract}
In this short note, we describe an approval-based committee selection rule 
that admits a polynomial-time algorithm and satisfies the Extended Justified Representation (EJR)
axiom. This rule is based on approximately maximizing the PAV score, by means of local search.
Our proof strategy is to show that
this rule provides almost optimal average satisfaction to all cohesive groups of voters, 
and that high average satisfaction for cohesive groups implies extended justified representation.
\end{abstract}

\section{Introduction}
Committee selection rules, i.e., rules that, given a collection of voters' preferences over candidates,
output a fixed-size set of winners (a committee), have received a considerable amount of attention in the last 
few years~\cite{FSST-trends}. For approval-based committee selection rules, i.e., rules where
each voter reports a subset of candidates that she approves of, an influential recent paper of 
Aziz et al.~\cite{aziz:scw} has proposed an axiom called {\em extended justified representation (EJR)}.
Informally, this axiom requires that if there is a large group of voters whose preferences
have  substantial overlap, then this group should be well-represented in the committee.
Aziz et al.~\cite{aziz:scw} have established that, among existing committee selection rules, 
Proportional Approval Voting (PAV) is the only rule that satisfies this axiom. Unfortunately, 
computing the output of this rule is NP-hard~\cite{AGG+14a}, and, moreover, Aziz et al.~\cite{aziz:scw}
have established that verifying if a given committee provides EJR is computationally hard as well.
Therefore, it was conjectured that finding committees that provide EJR is NP-hard.

In this paper, we show that this conjecture is false: we describe a polynomial-time procedure
that is guaranteed to output a committee that provides EJR. Our proof makes use of a recently introduced
notion of {\em average satisfaction} \cite{aaai/SEL17-pjr,DBLP:journals/corr/SkowronLBPE16}: we show that our rule offers
very high average satisfaction to all groups of voters that are large and cohesive, and deduce
from this that it satisfies EJR. Thus, in a sense, the welfare guarantees provided by our rule
are even much stronger than those offered by EJR.

\section{Preliminaries}
An {\em election} is a pair $E = (N, C)$, where $N = \{1, \ldots, n\}$ is a set of \emph{voters} and 
$C = \{c_1, \ldots, c_m\}$ is a set of \emph{candidates}. Voters in $N$ have {\em approval preferences}:
each voter in $N$ approves a subset of candidates in $C$. 
For each $i\in N$ we denote by $A_i \subseteq C$ the set of candidates approved by voter $i$.
We are interested in procedures that, given an election $E$ and a positive integer $k$, $1\le k\le |C|$,
output a non-empty collection of 
subsets of candidates $W\subseteq C$ of size exactly $k$; such procedures are called 
{\em committee selection rules}.

Given an election $E=(N, C)$, we define the PAV-score of a committee $W\subseteq C$ as
\begin{align*}
\pavscore(W) = \sum_{i = 1}^n \sum_{j = 1}^{|A_i \cap W|} \frac{1}{j} \text{.}
\end{align*}

{\em Proportional Approval Voting (PAV)} is the committee selection rule that, given an election
$E$ and a committee size $k$, outputs a committee of size $W$ with the highest PAV-score;
finding a committee in the output of this rule is NP-hard~\cite{AGG+14a,owaWinner}. 

We say that a group of voters $V \subseteq N$ is {\em $\ell$-cohesive} if 
$|V| \geq \left\lceil \ell \cdot \frac{n}{k} \right\rceil$ and $|\bigcap_{i \in V}A_i| \geq \ell$.

\begin{definition}[S{\'a}nchez-Fern{\'a}ndez et al.~\cite{aaai/SEL17-pjr}]
  Given a committee $W$, the \emph{average satisfaction of a group of voters $V\subseteq N$ 
  with respect to $W$} is defined as 
  \begin{align*}
  \aver{W}(V) = \frac{1}{|V|}\sum_{i\in V} |A_i \cap W|.
  \end{align*}
\end{definition}

\begin{definition}[Aziz~et~al.~\cite{aziz:scw}]
  A committee $W$ provides \emph{extended justified representation} (EJR) if for every $\ell>0$ 
  and every $\ell$-cohesive group of voters $V$
  there exists a voter $i \in V$ who approves at least $\ell$ members of $W$, i.e., $|A_i\cap W|\geq \ell$.
  A committee selection rule {\em satisfies extended justified representation} 
  if for every election $E=(N, C)$ and every $k$ with $1\le k\le |C|$ all committees that 
  it outputs provide EJR. 
\end{definition} 

S\'anchez-Fern\'andez et al.~\cite{aaai/SEL17-pjr} show that if a committee $W$ provides EJR  
then for every $\ell>0$ and every $\ell$-cohesive group of voters $V$ it holds that 
the average satisfaction of $V$ with respect to $W$ is at least $\frac{\ell-1}{2}$.
We observe that, conversely, if a committee provides very high average satisfaction
to all cohesive groups then it provides EJR.

\begin{lemma}\label{lem:avrep<->ejr}
Consider an election $E=(N, C)$ and a committee $W$.
If for every $\ell$-cohesive group it holds that its average satisfaction with respect to $W$ 
is strictly greater than $\ell-1$, then $W$ provides EJR.
\end{lemma}

\begin{proof}
Let $V$ be an $\ell$-cohesive group. Since the average satisfaction of voters in $V$ is greater than $\ell-1$, 
there is at least one voter $i\in V$ with $|A_i\cap W|\geq \ell$.
\qed
\end{proof}

\section{Local search Algorithm for PAV}

We define a new rule, $\ejrls$, which is a local search algorithm for PAV (see Figure~\ref{alg:local-search}). 

\begin{figure}[thb]
\begin{algorithm}[H]
   \DontPrintSemicolon
   %\SetAlCapFnt{\small}
   $W \leftarrow k$ \textrm{arbitrary candidates from} $C$\;
   \While{\textrm{there exist} $c \in W$ \textrm{and} $c' \in C \setminus W$ \textrm{such that} \\
         \quad\quad\quad $\pavscore\big((W \setminus \{c\}) \cup \{c'\}\big) \geq \pavscore(W) + \frac{n}{k^2}$}{
      $W \leftarrow (W \setminus \{c\}) \cup \{c'\}$\;
   }
   \Return{$W$}\;
\caption{$\ejrls$: a local search algorithm for PAV}\label{alg:local-search}
\end{algorithm}
\end{figure}

\begin{theorem}\label{thm:avrep}
Consider an election $E=(N, C)$ and a positive integer $k$, $1\le k \le |C|$.
Let $W$ be a winning committee chosen by $\ejrls$ on $(E, k)$.
Then for every $\ell>0$ and every $\ell$-cohesive group $V$ it holds that $\aver{W}(V)>\ell-1$, 
i.e., all $\ell$-cohesive groups have average satisfaction of more than $\ell-1$.
\end{theorem}
\begin{proof}
Assume for the sake of contradiction that for some $(E, k)$
$\ejrls$ outputs a committee $W$ such that there exists an $\ell$-cohesive group $V$ with $\aver{W}(V)\leq\ell-1$. 
Let $w_i = |W \cap A_i|$. 
%; in particular, we have $w_i \leq \ell - 1$ for all $i \in V$.

As $V$ is $\ell$-cohesive, there exist $\ell$ candidates approved by all voters in $V$.
At least one such candidate does not appear in $W$, since otherwise we would have 
$\aver{W}(V)\geq\ell$. Let $c$ be some such candidate.
Now consider a candidate $c' \in W$.
%, and assume that $c'$ is approved by voters $i_1, \ldots, i_p$ from $V$ and by voters $i_{p+1}, \ldots i_{r}$ from $N \setminus V$. 
If we remove $c'$ from the committee and add $c$ instead,
we increase the PAV-score of $W$ by 
\begin{align*}
\Delta(c, c') %&= \sum_{i \in V} \frac{1}{w_i + 1} - \sum_{j = 1}^r \frac{1}{w_{i_j}} + \sum_{j = 1}^p\left(\frac{1}{w_{i_j}} - \frac{1}{w_{i_j}+1}\right) \\
              &\ge \underbrace{\sum_{i \in V:c'\notin A_i} \frac{1}{w_i + 1} }_\text{adding $c$}  - \underbrace{\sum_{i\in N\setminus V\colon c' \in A_i} \frac{1}{w_i}}_\text{removing $c'$}  \\
              &= \sum_{i \in V} \frac{1}{w_i + 1} - \sum_{i\in N\colon c' \in A_i} \frac{1}{w_i} + \sum_{i \in V\colon c' \in A_i}\left(\frac{1}{w_i} - \frac{1}{w_i+1}\right).
\end{align*}
Note that $\Delta(c, c')$ may be negative for some $c'\in W$.
By the inequality between arithmetic and harmonic means we obtain
\begin{align}\label{eq:amhm}
\sum_{i \in V} \frac{1}{w_i + 1} \geq \frac{|V|^2}{\sum_{i \in V} (w_i +1)}= 
\frac{|V|}{\frac{\sum_{i \in V}w_i}{|V|}  + 1} = \frac{|V|}{\aver{W}(V)+1}\geq \frac{|V|}{\ell}.
\end{align}
Now, observe that
{\allowdisplaybreaks
\begin{align*}
\sum_{c' \in W} \Delta(c, c') 
&\ge \sum_{c' \in W} \Bigg(\sum_{i \in V} \frac{1}{w_i + 1} - \sum_{i\in N\colon c' \in A_i} \frac{1}{w_i} + 
   \sum_{i \in V\colon c' \in A_i}\left(\frac{1}{w_i} - \frac{1}{w_i+1}\right) \Bigg) \\
&= k \sum_{i \in V} \frac{1}{w_i + 1} - \sum_{c' \in W} \sum_{i\in N\colon c' \in A_i} \frac{1}{w_i} + 
     \sum_{c' \in W} \sum_{i \in V\colon c' \in A_i}\left(\frac{1}{w_i} - \frac{1}{w_i+1}\right)  \\
&= k \sum_{i \in V} \frac{1}{w_i + 1} - \sum_{i \in N} \sum_{c' \in W \cap A_i} \frac{1}{w_i} + 
     \sum_{i \in V} \sum_{c' \in W \cap A_i} \left(\frac{1}{w_i} - \frac{1}{w_i+1}\right) \\
&= k \sum_{i \in V} \frac{1}{w_i + 1} - n + |V| - \sum_{i \in V} \frac{w_i}{w_i+1} \\
                               %&\geq k\cdot  \frac{\sum_{i \in V} (w_i +1)}{|V|} - n + |V| - \sum_{i \in V} \frac{k-1}{k} \\
&= k \sum_{i \in V} \frac{1}{w_i + 1} - n + |V| - \sum_{i \in V} \left(1-\frac{1}{w_i+1}\right) \\
&= (k+1) \sum_{i \in V} \frac{1}{w_i + 1} - n.
\end{align*}
}
Further, by equation~\eqref{eq:amhm} we obtain
$$
k\sum_{i \in V} \frac{1}{w_i + 1} - n \ge \frac{k|V|}{\ell} - n 
\ge \left\lceil\frac{\ell n}{k}\right\rceil\cdot \frac{k}{\ell} - n \ge 0,
$$
and hence 
$$
\sum_{c' \in W} \Delta(c, c') \ge
(k+1)\sum_{i \in V} \frac{1}{w_i + 1} - n \ge \sum_{i \in V} \frac{1}{w_i + 1} \ge \frac{|V|}{\ell} \ge \frac{n}{k}.
$$

From the pigeonhole principle it follows that there exists $c' \in W$ such that $\Delta(c, c') \geq \frac{n}{k^2}$, 
which means that $W$ could not have been returned by our local search algorithm. This completes the proof.
\qed\end{proof}

\begin{corollary}
For PAV all $\ell$-cohesive groups have average satisfaction of more than $\ell-1$.
\end{corollary}

Note that Theorem~\ref{thm:avrep} applies not only to $\ejrls$ but also to PAV, as PAV selects a committee with 
the maximum PAV-score.

\begin{corollary}
$\ejrls$ satisfies extended justified representation.
\end{corollary}

\begin{proof}
Let $E=(N, C)$ be an election, let $k$ be a positive integer with $1\le k\le |C|$
and let $W$ be a winning committee chosen by $\ejrls$. 
Further, let $V$ be an $\ell$-cohesive group.
Then, by Theorem~\ref{thm:avrep} it holds that $\aver{W}(V)>\ell-1$.
Consequently, by Lemma~\ref{lem:avrep<->ejr}, $W$ provides EJR.
\qed\end{proof}

\begin{proposition}\label{prop:ls_running_time}
$\ejrls$ is computable in polynomial time.
\end{proposition}
\begin{proof}
A single improving swap can be clearly found and executed in polynomial time. Now, let us assess how many 
improvements the local search algorithm may perform. Each improvement increases the total PAV-score 
of a committee by at least $\frac{n}{k^2}$. 
The maximum score a committee can get is $n\cdot (1+\nicefrac 1 2 +\dots+\nicefrac 1 k)=O(n\ln(k))$. 
Thus, there may be at most $O(k^2\ln(k))$ improving swaps.
\qed\end{proof}

Observe that Proposition~\ref{prop:ls_running_time} relies on
having a threshold of $\frac{n}{k^2}$ in the definition of the local search algorithm. 
If we perform a swap each time when 
$\pavscore\big((W \setminus \{c\}) \cup \{c'\}\big) \geq \pavscore(W)$, 
this could potentially lead to a superpolynomial running time, 
since by a naive argument we could only conclude that each swap 
increases the total PAV score of a committee by one over the least common 
multiple of values $1, 2, \ldots, k$, which can be exponential with respect to $k$.

We will now show that the satisfaction guarantee given by Theorem~\ref{thm:avrep} cannot be improved.

\begin{figure}
\centering
\begin{subfigure}[t]{0.4\textwidth}
\centering
\begin{tikzpicture}[every node/.style={circle,draw=black, fill=black, inner sep=0pt,minimum size=3pt}]
    \node[
      regular polygon,
      regular polygon sides=4,
      minimum width=40mm,
      fill=white,draw=none
    ] (PG) {}
      (PG.corner 1) node[draw] (PG1) {}
      (PG.corner 2) node[draw] (PG2) {}
      (PG.corner 3) node[draw] (PG3) {}
      (PG.corner 4) node[draw] (PG4) {}
    ;
    \foreach \S/\E in {
      1/2, 2/3, 3/4, 4/1%
    } {
    \foreach \U in {
      0.33, 0.66
    } {
    \node[draw] at ($(PG\S)!\U!(PG\E)$) {};
    }};
    \draw[very thick,rounded corners=8,rotate=2*90,inner sep=20pt]($(PG1) - (.3,0.6)$)rectangle($(PG4) -(-.3,-0.6)$);
    \draw[very thick,rounded corners=8,rotate=3*90,inner sep=20pt]($(PG2) - (.3,0.6)$)rectangle($(PG1) -(-.3,-0.6)$);
    \draw[very thick,rounded corners=8,rotate=4*90,inner sep=20pt]($(PG3) - (.3,0.6)$)rectangle($(PG2) -(-.3,-0.6)$);
    \draw[very thick,rounded corners=8,rotate=5*90,inner sep=20pt]($(PG4) - (.3,0.6)$)rectangle($(PG3) -(-.3,-0.6)$);
  \end{tikzpicture}
    \caption{Example for $k=3$}
    \label{fig:k=4}
\end{subfigure}\quad%
\begin{subfigure}[t]{0.4\textwidth}
\centering
\begin{tikzpicture}[scale=1,every node/.style={circle,draw=black, fill=black, inner sep=0pt,minimum size=3pt}]
    \node[
      regular polygon,
      regular polygon sides=5,
      minimum width=40mm,
      fill=white,draw=none
    ] (PG) {}
      (PG.corner 1) node[draw] (PG1) {}
      (PG.corner 2) node[draw] (PG2) {}
      (PG.corner 3) node[draw] (PG3) {}
      (PG.corner 4) node[draw] (PG4) {}
      (PG.corner 5) node[draw] (PG5) {}
    ;
    \foreach \S/\E in {
      1/2, 2/3, 3/4, 4/5, 5/1%
    } {
    \foreach \U in {
      0.25, 0.5, 0.75
    } {
    \node[draw] at ($(PG\S)!\U!(PG\E)$) {};
    }};
    \draw[very thick,rounded corners=8,rotate=72,inner sep=20pt]($(PG5) + (.6,0.3)$)rectangle($(PG4) +(-.6,-0.3)$);
    \draw[very thick,rounded corners=8,rotate=2*72,inner sep=20pt]($(PG1) + (.6,0.3)$)rectangle($(PG5) +(-.6,-0.3)$);
    \draw[very thick,rounded corners=8,rotate=3*72,inner sep=20pt]($(PG2) + (.6,0.3)$)rectangle($(PG1) +(-.6,-0.3)$);
    \draw[very thick,rounded corners=8,rotate=4*72,inner sep=20pt]($(PG3) + (.6,0.3)$)rectangle($(PG2) +(-.6,-0.3)$);
    \draw[very thick,rounded corners=8,rotate=5*72,inner sep=20pt]($(PG4) + (.6,0.3)$)rectangle($(PG3) +(-.6,-0.3)$);
  \end{tikzpicture}
  \caption{Example for $k=4$}
  \label{fig:k=5}
\end{subfigure}%
 \caption{A visualization of the profiles used in Example~\ref{example}. Dots represent voters and boxes represent candidates; candidates are approved by those voters that the respective boxes contain.}
\end{figure}
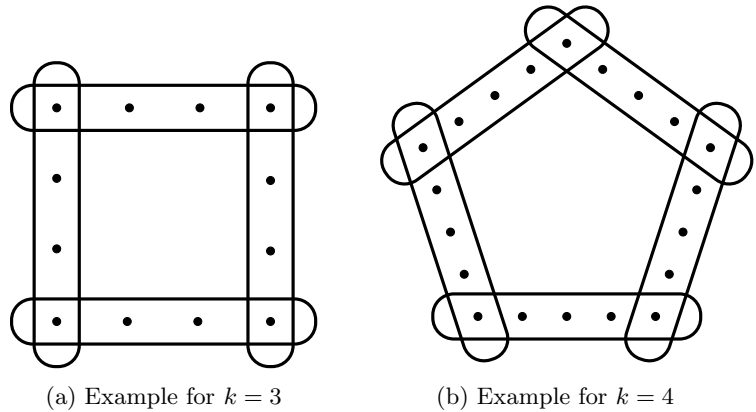

\begin{example}\label{example}
Consider the following election.
	\begin{align*}
		&1 \times \{ d,a\} & & 2 \times \{ a\} \\
		&1 \times \{ a,b\} & & 2 \times \{ b\} \\		
		&1 \times \{ b,c\} & & 2 \times \{ c\} \\
		&1 \times \{ c,d\} & & 2 \times \{ d\} 
	\end{align*}

This profile is schematically shown in Figure~\ref{fig:k=4}. For $k=3$, we have $\frac{n}{k}=4$ and 
consequently all voters that approve a fixed candidate form a $1$-cohesive group. The profile is symmetric with 
respect to candidates so without loss of generality assume that committee $\{a,b,c\}$ is chosen. Let us consider 
the voters who approve $d$, i.e., $\{d,a\}$, $2\times \{d\}$, and $\{c,d\}$. The average satisfaction of 
this group is $\nicefrac{1}{2}$. Hence we have found a profile where it is impossible to guarantee an average 
satisfaction for $1$-cohesive groups that is better than $\nicefrac{1}{2}$.

Let us now extend this example for $k=4$, moving from a rectangle shape to a pentagon (see Figure~\ref{fig:k=5}).
	\begin{align*}
		&1 \times \{ e,a\} & & 3 \times \{ a\} \\
		&1 \times \{ a,b\} & & 3 \times \{ b\} \\		
		&1 \times \{ b,c\} & & 3 \times \{ c\} \\
		&1 \times \{ c,d\} & & 3 \times \{ d\} \\
		&1 \times \{ d,e\} & & 3 \times \{ e\} 		
	\end{align*}

By the same argument as before, we can assume without loss of generality that $e$ is not contained in the 
committee. The average satisfaction of the respective group is $\frac{2}{5}$ and we see that we 
cannot guarantee an average satisfaction of more that $\frac25$ for $1$-cohesive groups. 
If we generalize these types of examples to larger $k$, we can deduce that we cannot guarantee an average
satisfaction of more than $\frac{2}{k}$ for $1$-cohesive groups. Hence, for arbitrary $k$, there is no 
positive constant $\gamma$ such that we can guarantee an average satisfaction of $\gamma$ for 
$1$-cohesive groups.
\end{example}

\begin{proposition}
Let $\ell$ be a positive integer and $\gamma>0$.
It is not possible to guarantee an average satisfaction of $\ell-1+\gamma$ for $\ell$-cohesive groups.
\end{proposition}

\begin{proof}
We have seen in Example~\ref{example} how to construct profiles where $1$-cohesive groups 
cannot have an average satisfaction better than $\frac{2}{k}$. 
If we choose $k>\frac 2 \gamma$, then we have shown the statement for $\ell=1$.
This construction can easily be generalized for arbitrary $\ell$. 
We replace each candidate with $\ell$ copies; voters approve all copies of previously approved candidates. 
In this case there is at least one $\ell$-cohesive group with only $\ell-1$ 
of their joint candidates approved; let this group be $V$. 
We have $\aver{W}(V)=\ell-1+\frac 2 \gamma$. As before, for $k>\frac 2 \gamma$ 
we obtain an example showing that an average satisfaction of $\ell-1+\gamma$ 
for $\ell$-cohesive groups cannot be guaranteed.
\qed\end{proof}

%\subsubsection*{Acknowledgements.} We thank ...

\bibliographystyle{splncs03}
\bibliography{main}

\begin{thebibliography}{1}
\providecommand{\url}[1]{\texttt{#1}}
\providecommand{\urlprefix}{URL }

\bibitem{aziz:scw}
Aziz, H., Brill, M., Conitzer, V., Elkind, E., Freeman, R., Walsh, T.:
  Justified representation in approval-based committee voting. Social Choice
  and Welfare  48(2),  461--485 (2017)

\bibitem{AGG+14a}
Aziz, H., Gaspers, S., Gudmundsson, J., Mackenzie, S., Mattei, N., Walsh, T.:
  Computational aspects of multi-winner approval voting. In: 14th International
  Conference on Autonomous Agents and Multiagent Systems ({AAMAS}). pp.
  107--115 (2015)

\bibitem{FSST-trends}
Faliszewski, P., Skowron, P., Slinko, A., Talmon, N.: Multiwinner voting: A new
  challenge for social choice theory. In: Endriss, U. (ed.) Trends in
  Computational Social Choice. AI Access (2017)

\bibitem{aaai/SEL17-pjr}
S{\'a}nchez-Fern{\'a}ndez, L., Elkind, E., Lackner, M., Fern{\'a}ndez, N.,
  Fisteus, J.A., {Basanta Val}, P., Skowron, P.: Proportional justified
  representation. In: Proceedings of the 31st AAAI Conference on Artificial
  Intelligence (AAAI 2017). pp. 670--676. AAAI Press (2017)

\bibitem{owaWinner}
Skowron, P., Faliszewski, P., Lang, J.: Finding a collective set of items: From
  proportional multirepresentation to group recommendation. Artificial
  Intelligence  241,  191--216 (2016)

\bibitem{DBLP:journals/corr/SkowronLBPE16}
Skowron, P., Lackner, M., Brill, M., Peters, D., Elkind, E.: Proportional
  rankings. CoRR  abs/1612.01434 (2016), \url{http://arxiv.org/abs/1612.01434}

\end{thebibliography}

\end{document}